\documentclass[11pt,a4paper]{article}

\usepackage[
  top=25mm,
  bottom=25mm,
  left=25mm,
  right=25mm
]{geometry}

\usepackage{microtype}
\usepackage{xcolor}
\usepackage{authblk}

\usepackage{amsmath}
\usepackage{amsfonts}
\usepackage{amssymb}
\usepackage{amsthm}
\usepackage{mathtools}
\usepackage{bm}
\usepackage{bbm}
\usepackage{braket}

\usepackage{graphicx}
\usepackage{subcaption}
\usepackage{booktabs}
\usepackage{array}
\usepackage{multirow}
\usepackage{algorithm}
\usepackage{algpseudocode}

\usepackage[
    backend=biber,
    style=phys,
    sorting=none,
    biblabel=brackets,
    articletitle=true,
    pageranges=false,
    maxbibnames=10,
    minbibnames=10,
    eprint=true,
    doi=false,
    url=false
  ]{biblatex}

\DeclareFieldFormat[
article,inproceedings,misc,unpublished
]{title}{#1}

\DeclareFieldFormat[misc,unpublished]{date}{%
\iffieldundef{eprint}{#1}{\mkbibparens{#1}}%
}

\AtEveryBibitem{%
    \iffieldundef{journaltitle}{%
      \iffieldundef{booktitle}{}{%
        \clearfield{eprint}%
        \clearfield{eprinttype}%
        \clearfield{eprintclass}%
      }%
    }{%
      \iffieldequalstr{journaltitle}{arXiv}{}{%
        \clearfield{eprint}%
        \clearfield{eprinttype}%
        \clearfield{eprintclass}%
      }%
    }%
    
    \iffieldundef{eprint}{}{%
      \clearfield{month}%
      \clearfield{day}%
      \clearfield{eprintclass}%
    }%
}

\usepackage[
colorlinks=true,
linkcolor=black,
citecolor=blue!50!black,
urlcolor=blue!50!black
]{hyperref}

\newtheorem{theorem}{Theorem}[section]
\newtheorem{proposition}[theorem]{Proposition}

\title{Classical Simulation and Design Frontiers for IBM's Doped Clifford Sampling Experiment}
\author[1]{Hidetaka Manabe\thanks{Email: \href{mailto:hidetaka_manabe@sutd.edu.sg}{\nolinkurl{hidetaka_manabe@sutd.edu.sg}}}}
\author[2]{Hanfeng Gu}
\author[1]{Feng Pan\thanks{Email: \href{mailto:feng_pan@sutd.edu.sg}{\nolinkurl{feng_pan@sutd.edu.sg}}}}
\affil[1]{Science, Mathematics and Technology Cluster, Singapore University of Technology and Design (SUTD), Singapore}
\affil[2]{NVIDIA Corporation}

\date{}

\begin{document}

\maketitle

\begin{abstract}
We classically simulate the IBM doped Clifford random circuit sampling experiment, comprising $70$ qubits, $70$ entangling layers, and $468$ inserted $T$ gates, as reported in the first version of Ref.~\cite{martielSamplingHardCircuits2026}.
A deterministic temporal-boundary tensor network contraction approach is specifically designed to tackle such open-boundary one-dimensional brickwork circuits with operator-Schmidt-rank-$2$ entangling gates.
For an $n$-qubit circuit of depth $d$, the resulting unsliced path evaluates an exact amplitude with contraction width $\lceil d/2\rceil$; Ratcatcher calculations certify that no smaller width is possible for the tested instances.
Because one-qubit gates are absorbed without changing the network topology, the width and dense scheduled contraction cost are independent of their values and of the number and placement of $T$ gates.
For this initial instance, its largest intermediate tensor contains $2^{35}$ complex64 entries (256 times smaller than IBM's estimation), corresponding to a tensor payload of $256$ GiB, and is distributed across eight GPUs within a node.
Using 32 nodes, with eight NVIDIA H100 GPUs per node, we completed all 2051 amplitude batches corresponding to IBM's published output bitstrings in 37.3 minutes.
The resulting probabilities yield a log-XEB estimate of $0.35034$ with a 95\% interval of $[0.29763,0.40305]$.
Under the Porter--Thomas and scrambled-noise assumptions, this is numerically compatible with IBM's fidelity lower bound; separately, fidelity-weighted resource accounting projects a 583-contraction workload with a 10.6-minute makespan on the same 32 nodes.
This simulation motivated a revised $64$-qubit experiment with periodic boundary conditions.
Our analysis indicates that this revised instance lies beyond the practical reach of the exact tensor network contraction methods considered here.
More broadly, the approach provides a practical diagnostic for experimental outputs and a quantitative tool for designing future doped Clifford sampling experiments.
\end{abstract}

\section{Introduction}
\label{sec:introduction}

Quantum computing offers proven algorithmic speedups for problems such as factoring~\cite{shorPolynomialTimeAlgorithmsPrime1997} and unstructured search~\cite{groverFastQuantumMechanical1996}, and it motivates potential advantages in quantum many-body simulation~\cite{feynmanSimulatingPhysicsComputers1982a} and optimization~\cite{farhiQuantumApproximateOptimization2014}.
Realizing these applications at a practically useful scale, however, is expected to require fault-tolerant quantum computers with many logical qubits and a correspondingly larger number of physical qubits~\cite{gidneyHowFactor20482021,beverlandAssessingRequirementsScale2022b}. 
Before such machines are available, an important intermediate objective is to demonstrate a computation that can be executed on current quantum hardware but cannot be reproduced within comparable resources on a classical computer~\cite{harrowQuantumComputationalSupremacy2017}.

Random circuit sampling (RCS) has become a principal setting for such quantum-advantage experiments, supported by complexity-theoretic evidence, anticoncentration, and a growing body of experimental work~\cite{aaronsonComplexityTheoreticFoundationsQuantum2016,boulandComplexityVerificationQuantum2019,hangleiterComputationalAdvantageQuantum2023,boixoCharacterizingQuantumSupremacy2018,aruteQuantumSupremacyUsing2019,wuStrongQuantumComputational2021,zhuQuantumComputationalAdvantage2022,morvanPhaseTransitionRandom2023,gaoEstablishingNewBenchmark2024}.
The task is not intended as an end-user application; rather, it provides a controlled benchmark in which quantum hardware, complexity-theoretic evidence, and classical simulation can be compared. 
Classical simulation has two roles in this comparison.  
It establishes the resource baseline that the quantum experiment must exceed~\cite{zlokapaBoundariesQuantumSupremacy2023a,huangEfficientParallelizationTensor2021}.
When exact amplitudes can be computed, it also provides an independent diagnostic of the measured output distribution under suitable assumptions~\cite{villalongaFlexibleHighperformanceSimulator2019,villalongaEstablishingQuantumSupremacy2019,liuVerifyingQuantumAdvantage2024,gaoLimitationsLinearCrossEntropy2024,morvanPhaseTransitionRandom2023}.

Martiel \emph{et al.} recently demonstrated a particularly important RCS experiment based on a doped Clifford circuit and spacetime-code postselection~\cite{martielSamplingHardCircuits2026}.  
Their hard instance reported in the first version of their paper contains $n=70$ qubits on a one-dimensional open chain, $d=70$ alternating nearest-neighbor CZ layers, and 468 inserted $T$ gates. 
The logical circuit is embedded in a 97-qubit experiment, and syndromes obtained from a Clifford reference are used to detect errors and postselect the output~\cite{delfosseSpacetimeCodesClifford2023a,martielLowoverheadErrorDetection}.
The experiment produced 2051 postselected samples in 16.1 minutes and reported, at 95\% confidence, a state-fidelity lower bound of $0.284$ for the hard circuit.  
Although the fidelity certificate and XEB are not identical estimands, this value is far above the order-$10^{-3}$ XEB fidelities reported for the largest circuits in early RCS advantage demonstrations~\cite{aruteQuantumSupremacyUsing2019,zlokapaBoundariesQuantumSupremacy2023a}.

The doped Clifford construction is designed to frustrate several important classical approaches, whose costs are controlled by different resource parameters. 
State-vector methods scale with the Hilbert-space dimension $2^n$~\cite{deraedtMassivelyParallelQuantum2007,haner05PetabyteSimulation2017a,raedtMassivelyParallelQuantum2019}; MPS and related low-rank methods are governed by entanglement and the tolerated approximation error~\cite{zhouWhatLimitsSimulation2020}; and stabilizer or near-Clifford methods are governed by magic, $T$ count, or the size of an active non-Clifford subspace~\cite{bravyiImprovedClassicalSimulation2016,bravyiSimulationQuantumCircuits2019,kissingerSimulatingQuantumCircuits2022,melloHybridStabilizerMatrix2024,liuClassicalSimulabilityClifford+T2024,chaseClifftFastExact2026}.
Noise can make one-dimensional RCS efficiently simulable in regimes where the ideal distribution remains hard, and polynomial-time asymptotic algorithms are known for noisy anticoncentrated random circuits~\cite{nohEfficientClassicalSimulation2020,aharonovPolynomialtimeClassicalAlgorithm2022}.
However, they estimate in Ref.~\cite{martielSamplingHardCircuits2026} that the effective noise level of their experiment is sufficiently low that such noise-based methods are expected to remain impractical.

Nevertheless, the circuit also has a structural feature favorable to exact tensor network simulation~\cite{markovSimulatingQuantumComputation2008}: every entangling gate is a CZ gate and therefore has operator Schmidt rank two. 
Accordingly, each CZ gate decomposes into two rank-three tensors sharing a two-dimensional internal index.
After one-qubit gates and boundary vectors are absorbed, the circuit becomes a rectangular planar tensor network whose short direction is approximately half the circuit depth.
This geometry suggests that, instead of evolving a $2^n$ component state vector forward in circuit time, we can evolve a temporal boundary state in the spatial direction.
Although such transverse contraction is a standard tensor network idea~\cite{banulsMatrixProductStates2009a,leroseInfluenceMatrixApproach2021}, the regular geometry of this circuit family permits a deterministic sweep with minimum contraction width.
Previous tensor network simulators have exploited multi-amplitude and big-batch contractions to verify RCS experiments~\cite{liuVerifyingQuantumAdvantage2024,panSimulationQuantumCircuits2022b,panSolvingSamplingProblem2022a}.
Our construction specializes this idea to the open brickwork geometry: it gives an explicit deterministic schedule, certifies its width on the tested instances, and identifies a range of open output indices that does not increase that width.

The IBM supplement also reports a cotengra HyperOptimizer study of the same 70-qubit circuit, with a maximum intermediate of $2^{43}$ scalars in its unconstrained search and $2^{30}$ scalars when memory is constrained, at higher contraction cost~\cite{martielSamplingHardCircuits2026}.
The deterministic representation used here has width $35$ for the IBM instance.
Under matched tensorization and memory-accounting conventions, this eight-bit gap would correspond to a factor of $256$ in the largest-intermediate entry count; because the two studies do not document all conventions identically, we use the comparison to motivate a controlled path study rather than claiming a direct factor-$256$ hardware-memory reduction.

Our main contributions are as follows:
\begin{itemize}
  \item We construct a deterministic temporal-state contraction for an $n$-qubit open-boundary one-dimensional brickwork circuit with bond-dimension-$2$ entangling gates.
  For $n\geq4$ and $7\leq d\leq2n$, the path evaluates one exact amplitude with contraction width $\lceil d/2\rceil$, time $O(nd\,2^{\lceil d/2\rceil})$, and space $O(2^{\lceil d/2\rceil})$.
  Because the one-qubit gates are absorbed without changing the network topology, the width and dense scheduled contraction cost do not depend on their values or on the number and placement of $T$ gates.

  \item For the $n=70$ networks at $d=10,20,\ldots,70$ shown in Fig.~\ref{fig:depth-cost-and-width}, Ratcatcher calculations~\cite{seymourCallRoutingRatcatcher1994,ogormanParameterizationTensorNetwork2019,jakes-schauerCarvingwidthContractionTrees2019} exclude every width below that of the constructive path.
  In particular, the $n=d=70$ IBM network admits a certified width-$35$ path, whereas the generic heuristic search used in our comparison fails to recover this regular structure at large depth.

  \item We extend the construction to a batch with $k$ open output bits.
  For $1\leq k\leq\lfloor d/4\rfloor$, the contraction returns $2^k$ amplitudes while retaining width $\lceil d/2\rceil$.
  For the IBM calculation we use $k=8$, obtaining 256 amplitudes per contraction.

  \item We execute the width-$35$ path on the published IBM instance.
  Each contraction distributes its 256-GiB largest-intermediate tensor payload over eight H100 GPUs.
  Running 32 independent eight-GPU jobs in parallel, we evaluate all 2051 published-output batches in a measured makespan of 37.3 minutes.
  This measured calculation is distinct from the fidelity-weighted resource projection, which assigns 583 contractions and an estimated 10.6-minute makespan to the corresponding workload.

  \item From the exact probabilities of the published outputs, we obtain a log-XEB estimate of $0.35034$ with a 95\% interval $[0.29763,0.40305]$ when the outputs are treated as independent draws.
  Under the Porter--Thomas and scrambled-noise assumptions used to interpret log-XEB as a fidelity proxy, this value is numerically compatible with IBM's independently reported fidelity lower bound.

  \item We construct and analyze paths for open bond-dimension-$4$ brickwork circuits and periodic bond-dimension-$2$ circuits, and quantify slicing overhead for the analyzed strategies.
  These modeled sensitivities show how entangling-gate bond dimension, spatial boundary conditions, depth, and available memory affect the cost of this exact-amplitude approach.
\end{itemize}

Together, these results expose a separation between the features that make the IBM circuit difficult for state-vector, entanglement-based, and near-Clifford simulators and those that determine the cost of the present tensor network method.
The experiment has substantial entanglement and 468 non-Clifford gates, but its open quasi-one-dimensional geometry and operator-Schmidt-rank-$2$ entanglers admit a low-width transverse contraction.
This structure enables both a large-scale classical diagnostic of the published data and a quantitative study of circuit features that remove the favorable contraction geometry~\cite{decrossComputationalPowerRandom2024}.
Indeed, our independent classical simulation motivated a new experiment with a closed spatial boundary~\cite{martielSamplingHardCircuits2026}.
The updated circuit has $64$ qubits, $73$ entangling layers, $314$ inserted $T$ gates on a one-dimensional ring.
Our cost estimates indicate that classical simulation of the revised circuit is infeasible with the exact tensor network methods considered here.

The paper is organized as follows.
Section~\ref{sec:brickwork-circuits} introduces the one-dimensional brickwork circuit family and its tensor network representation.
Section~\ref{sec:exact-amplitude-contraction} presents the constructive path for exact amplitude evaluation and analyzes its contraction width.
Section~\ref{sec:amplitudes-to-samples} extends the method from amplitude evaluation to the sampling problem.
Section~\ref{sec:numerical-results} reports the single- and multi-GPU numerical results, including the simulation of IBM's circuit.
Section~\ref{sec:circuit-design} examines slicing, bond dimension, and boundary geometry and combines the resulting estimates in a classical-simulatability map.
Finally, Section~\ref{sec:discussion} summarizes our contributions and discusses future directions.

\section{Open-boundary 1D brickwork circuits}
\label{sec:brickwork-circuits}
\subsection{Bond-dimension-$\chi$ brickwork circuits}
\begin{figure}[t]
  \centering

  \begin{subfigure}[c]{0.66\textwidth}
      \centering
      \includegraphics[
          height=5.0cm,
          keepaspectratio
      ]{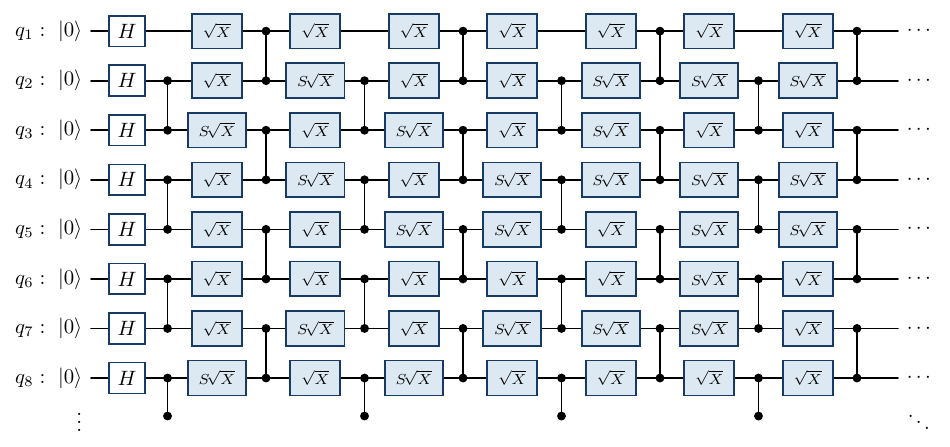}
      \caption{}
      \label{fig:circuit_fragment}
  \end{subfigure}
  \hfill
  \begin{subfigure}[c]{0.33\textwidth}
      \centering
      \includegraphics[
          height=5.5cm,
          keepaspectratio
      ]{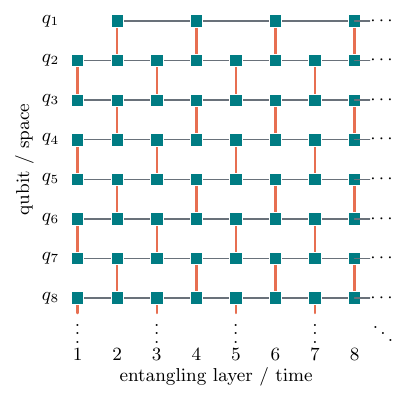}
      \caption{}
      \label{fig:peps_like_tn}
  \end{subfigure}

  \caption{
  Mapping an open-boundary one-dimensional brickwork circuit to its
  tensor network representation.
  (a) A layout of the brickwork Clifford circuit~\cite{martielSamplingHardCircuits2026}.
  (b) The corresponding PEPS-like tensor network geometry.
  }
  \label{fig:circuit_to_tn}
\end{figure}

We consider the circuit family introduced in the IBM doped Clifford sampling experiment~\cite{martielSamplingHardCircuits2026}.
The initial instance contains $n=70$ qubits on an open one-dimensional chain, $d=70$ entangling layers of nearest-neighbor $\mathrm{CZ}$ gates, and $468$ inserted $T$ gates.
Random single-qubit Clifford gates are applied between successive entangling layers.

Each entangling layer consists of non-overlapping nearest-neighbor $\mathrm{CZ}$ gates.
Successive layers alternate between the two brickwork patterns
\begin{equation}
    (2,3),(4,5),\dotsc \qquad \mathrm{and} \qquad (1,2),(3,4),\dotsc.
\end{equation}
Thus, each entangling layer contains approximately $n/2$ $\mathrm{CZ}$ gates.
The $n=70$ circuit alternates between layers of $34$ and $35$ $\mathrm{CZ}$ gates, for a total of $2,415$ $\mathrm{CZ}$ gates at $d=70$.
A fragment of the published circuit is shown in Fig.~\ref{fig:circuit_to_tn}(a).

More generally, we allow arbitrary one-qubit gates before, between, and after the entangling layers, and replace the $\mathrm{CZ}$ gates by nearest-neighbor two-qubit gates $U_{\mathrm{ent}}$.
For a two-qubit operator $U_{\mathrm{ent}}$, its operator Schmidt rank $\chi$ is the smallest $\chi$ for which
\begin{equation}
    U_{\mathrm{ent}} = \sum_{\alpha=1}^{\chi} A_\alpha \otimes B_\alpha,
\end{equation}
where $A_\alpha$ and $B_\alpha$ are single-qubit operators.
We refer to a bond-dimension-$\chi$ brickwork family when every two-qubit gate has operator Schmidt rank at most $\chi$.
For example, $\mathrm{CZ}$ and $\mathrm{CNOT}$ have $\chi=2$, whereas a generic $\mathrm{fSim}$ gate~\cite{aruteQuantumSupremacyUsing2019} and a generic two-qubit unitary gate have $\chi=4$.

\subsection{Tensor network representation of brickwork circuits}
We map the quantum circuit to a tensor network as illustrated in Fig.~\ref{fig:circuit_to_tn}(b).
Using the operator Schmidt decomposition, each two-qubit gate is replaced by a pair of rank-3 tensors, each with index dimensions $(2, 2, \chi)$, connected by an internal bond of dimension $\chi$.
The input states, output projectors, and adjacent one-qubit gates are absorbed into the neighboring tensors.
The resulting closed tensor network has a rectangular Projected Entangled Pair States (PEPS)-like geometry~\cite{verstraeteMatrixProductStates2008}, with spatial length $n$ and temporal width $\left\lceil d/2\right\rceil$.

Contracting the above tensor network yields an amplitude
\begin{equation}
    \psi(x)=\braket{x|U|0^n},
\end{equation}
and hence the output probability $p(x)=|\psi(x)|^2$ of getting the bitstring $x$.
The simulation cost analysis in \cite{martielSamplingHardCircuits2026} considers the task of evaluating such outcome probabilities.
Therefore, we can identify the cost of contracting the amplitude tensor network with the classical simulation cost of quantum circuits.

The computational cost depends strongly on the contraction path, that is, the order in which pairs of tensors are contracted.
We characterize the cost of a contraction path $P$ using the following quantities~\cite{grayHyperoptimizedTensorNetwork2021}:
\begin{itemize}
    \item \textbf{Contraction cost.}
    Consider a pairwise contraction
    \begin{equation}
        A_{IK}B_{KJ}\longrightarrow C_{IJ},
    \end{equation}
    where $K$ denotes the contracted multi-index and $I$ and $J$ are the retained multi-indices.
    The arithmetic cost of this step is calculated by $|I||J||K|$, where $|I|$ denotes the product of the dimensions of the indices in the multi-index $I$.
    We denote by $C(P)$ the total contraction cost obtained by summing this quantity over all contractions along the path.

    \item \textbf{Contraction width.}
    The contraction width is defined as
    \begin{equation}
        W(P) = \max_t \log_2 |T_t|,
    \end{equation}
    where $T_t$ denotes an intermediate tensor generated along the path and $|T_t|$ is its number of entries.
    This metric corresponds to the size of the largest intermediate tensor.

    \item \textbf{Peak memory usage.}
    The contraction width is a convenient proxy for the memory requirement, but it does not represent the actual peak memory usage.
    In practice, multiple tensors, auxiliary workspaces, and implementation-specific buffers may coexist in memory.
    We define the peak memory usage as the maximum amount of memory occupied at any point during the contraction.

    \item \textbf{Slicing overhead.}
    Slicing fixes selected indices and contracts each resulting (and smaller) tensor network independently~\cite{villalongaEstablishingQuantumSupremacy2019,huangEfficientParallelizationTensor2021}.
    Slicing a multi-index $S$ generates $|S|$ independent contraction tasks.
    By carefully choosing sliced indices, we may reduce the contraction width of each task, albeit generally increasing the total contraction cost.
    We define the slicing overhead as the ratio of the total contraction cost over all sliced tasks to that of the corresponding unsliced contraction.
\end{itemize}

\section{Exact amplitude contraction}
\label{sec:exact-amplitude-contraction}
In this section, we analyze the cost of contracting a tensor network to evaluate a single amplitude $\braket{x|U|0^n}$.
In what follows, we assume that all entangling gates have operator Schmidt rank $\chi=2$, as is the case for the $\mathrm{CZ}$ gates used in~\cite{martielSamplingHardCircuits2026}.
Here and below, ``exact'' means that the tensor network is contracted without bond truncation or an approximate gate decomposition; the numerical GPU calculations are performed in finite-precision complex arithmetic.
For a tensor $T$, we define $\operatorname{rank}(T)$ as the number of its legs.

\subsection{Constructive path for amplitude evaluation}
\label{sec:temporal-sweep}
We consider the rectangular PEPS-like tensor network with temporal width $\left\lceil d/2\right\rceil$ and spatial height $n$.
When $\left\lceil d/2\right\rceil \le n$, a natural candidate for the contraction path is a spatial sweep from one end of the qubit chain.
First, all tensors associated with the boundary qubit are contracted along the temporal direction and form a boundary tensor.
The contraction then proceeds along the spatial direction by successively absorbing the tensors.
The resulting algorithm can be viewed as evolving a temporal state vector along the spatial direction.

\begin{algorithm}[t]
  \caption{Deterministic temporal-state contraction for an exact amplitude}
  \label{alg:constructive-sweep}
  \begin{algorithmic}[1]
  \Require A brickwork tensor network $\mathcal{T}_x$
  representing $\braket{x|U|0^n}$
  \Ensure The exact amplitude $\braket{x|U|0^n}$

  \State $h\gets\lceil d/2\rceil$
  \State $\pi\gets\Call{FixedSweepOrder}{\mathcal{T}_x}$
  \Comment{$q_1$ to $q_n$; reverse time on even-numbered rows}
  \State $B\gets\pi[1]$
  \State $i\gets2$

  \While{$i\leq|\pi|$}
      \State $X\gets\pi[i]$

      \If{$\Call{RankAfterContract}{B,X}\leq h$}
          \State $B\gets\Call{Contract}{B,X}$
          \State $i\gets i+1$
      \Else
          \State $Y\gets\pi[i+1]$
          \State $C\gets\Call{Contract}{X,Y}$
          \State $B\gets\Call{Contract}{B,C}$
          \State $i\gets i+2$
      \EndIf
  \EndWhile

  \State \Return $B$
  \end{algorithmic}
\end{algorithm}

\begin{figure}[t]
  \centering
  \includegraphics[width=0.7\columnwidth]
  {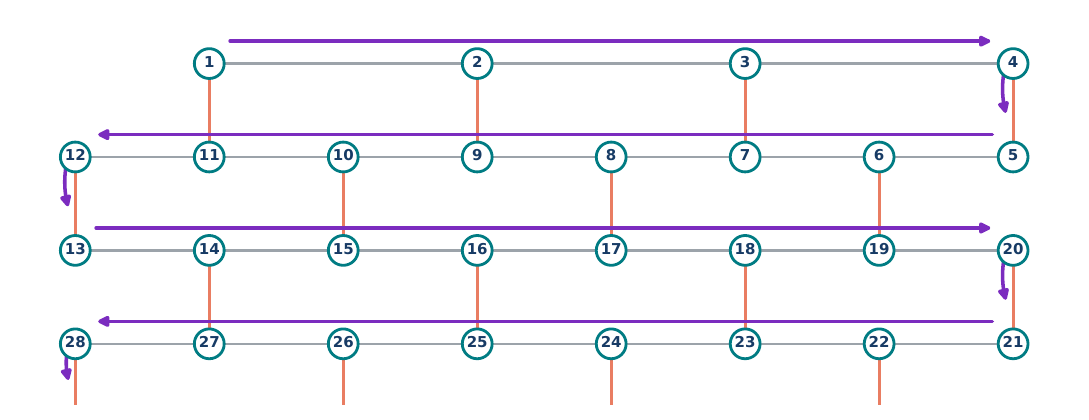}
  \caption{
      Deterministic contraction path for exact amplitude contraction.
      The numbers indicate the order in which tensors are visited, and the arrows indicate the direction of the sweep.
      $d=8$ in this figure.
  }
  \label{fig:constructive-sweep}
\end{figure}

We construct such a path deterministically, as summarized in Algorithm~\ref{alg:constructive-sweep} and illustrated in Fig.~\ref{fig:constructive-sweep}.

\begin{proposition}[Contraction width of the constructive sweep]
\label{prop:constructive-memory}
Consider an $n$-qubit bond-dimension-$2$ open-boundary brickwork tensor network with $n\geq 4,\,7\leq d\leq 2n$.
The contraction order in Algorithm~\ref{alg:constructive-sweep} produces an unsliced contraction path $P_{\mathrm{con}}$ whose contraction width is
\[
  W(P_{\mathrm{con}}) = \lceil d/2\rceil.
\]
\end{proposition}

\begin{proof}
Let $h=\lceil d/2\rceil$.
We proceed along the fixed tensor order shown in Fig.~\ref{fig:constructive-sweep}.
A single tensor is absorbed directly only when the resulting boundary tensor has rank at most $h$.
Consider a step in which the next two tensors $X$ and $Y$ are contracted first.
Since each has rank at most $3$ and they share one leg, their contraction $C$ has rank at most $4$.

\begin{figure}[t]
  \centering
  \begin{subfigure}[t]{0.48\columnwidth}
      \centering
      \includegraphics[width=\linewidth]
      {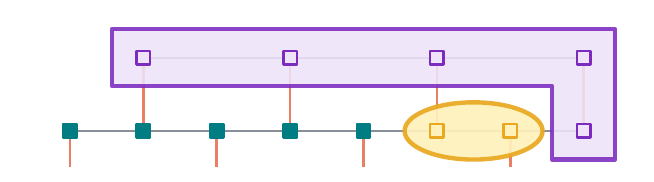}
      \caption{$d=8$: two shared legs.}
      \label{fig:paired-two-shared}
  \end{subfigure}
  \hfill
  \begin{subfigure}[t]{0.48\columnwidth}
      \centering
      \includegraphics[width=\linewidth]
      {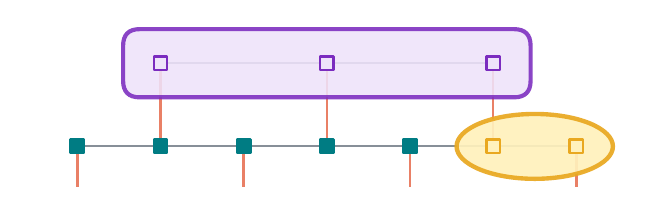}
      \caption{$d=7$: one shared leg at an odd-depth turn.}
      \label{fig:paired-one-shared}
  \end{subfigure}
  \caption{
      The two possible paired steps.
      The purple region contains the tensors already absorbed into the boundary tensor, and the yellow region contains the next two tensors
      to be contracted.
  }
  \label{fig:paired-update-cases}
\end{figure}

In the case shown in
Fig.~\ref{fig:paired-update-cases}(a), $C$ shares two legs with the boundary tensor $B$. Hence
\begin{equation}
    \operatorname{rank}\!\bigl(\operatorname{Contract}(B,C)\bigr)
  \leq \operatorname{rank}(B)+4-2\times2
  \leq h.
\end{equation}
When $d=2h$ is even, the bound is achieved after the first row has been contracted.

The only remaining case occurs at a row turn when $d=2h-1$ is odd, as shown in Fig.~\ref{fig:paired-update-cases}(b).
Before this contraction, $B$ contains all tensors in the preceding rows.
Its legs have
\begin{equation}
  \operatorname{rank}(B)=h-1.
\end{equation}
Here $C$ has rank three and shares one leg with $B$, so
\[
  \operatorname{rank}\!\bigl(\operatorname{Contract}(B,C)\bigr)
  =(h-1)+3-2=h.
\]
\end{proof}

The constructive path performs $O(nd)$ contraction steps.
Each step requires $O(2^h)$ arithmetic operations.
The total arithmetic cost is consequently
\[
  C(P)
  =
  O\!\left(nd\,2^h\right)
  =
  O\!\left(nd\,2^{\lceil d/2\rceil}\right).
\]

\subsection{Comparison with heuristic path search and exact width certification}
\label{sec:heuristic-comparison}

We compare the constructive sweep with contraction paths obtained using the cotengra HyperOptimizer~\cite{grayHyperoptimizedTensorNetwork2021}.
The HyperOptimizer searches were performed with the objective
\texttt{minimize="flops"}, which empirically produced the most competitive paths in terms of both contraction cost and width.
We utilized simulated annealing and subtree reconfiguration, and slicing and memory constraints were disabled.
The search was terminated after 128 consecutive trials without an improvement.

\begin{figure}[t]
    \centering
    \includegraphics[width=\linewidth]
    {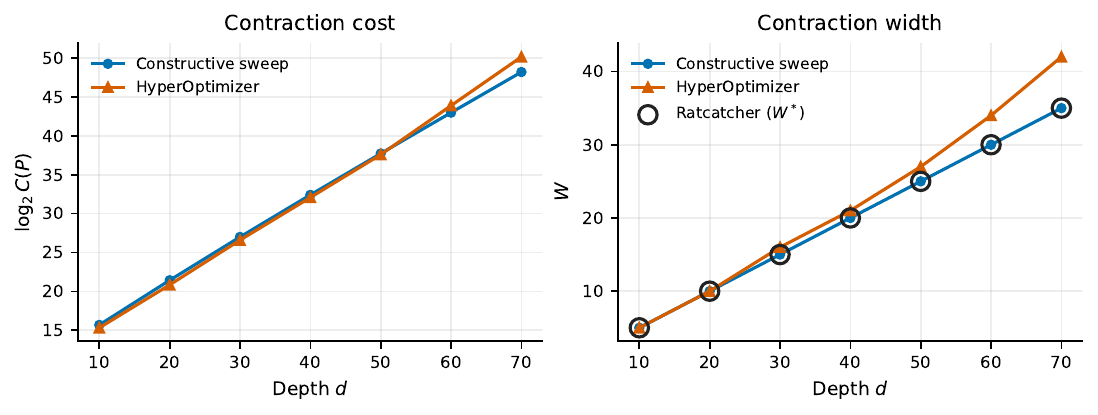}
    \caption{
        Total contraction cost (left) and contraction width (right) for the $n=70$ brickwork tensor networks.
        The blue lines show the constructive path, while the orange lines show the paths obtained by HyperOptimizer.
        Black open circles indicate the exact minimum contraction widths $W^*$ certified using Ratcatcher.
      }
    \label{fig:depth-cost-and-width}
\end{figure}

Figure~\ref{fig:depth-cost-and-width} shows the total contraction cost and contraction width as functions of the circuit depth.
The constructive sweep exactly attains $W(P_{\mathrm{con}})=\left\lceil d/2\right\rceil$,
while its arithmetic cost follows the bound $C(P_{\mathrm{con}})= O\!\left(nd\,2^{\lceil d/2\rceil}\right)$.
At small and intermediate depths, HyperOptimizer sometimes finds paths with a moderately smaller contraction cost.
However, the hyper-optimized paths begin to exceed the constructive width as the depth increases.
This deviation becomes progressively larger, and at the largest depths $d=70$ the returned paths are substantially worse in both contraction cost and contraction width.
This illustrates the increasing difficulty of recovering the regular structure using a heuristic search.

The planarity of the tensor networks also permits an independent certification of their minimum contraction width.
Contraction trees of a tensor network can be related to carving decompositions of its weighted graph~\cite{ogormanParameterizationTensorNetwork2019,jakes-schauerCarvingwidthContractionTrees2019}.
For planar graphs, the Ratcatcher algorithm decides in polynomial time whether a carving decomposition below a prescribed width exists ~\cite{seymourCallRoutingRatcatcher1994}.
We use the constructive width as an upper-bound witness and apply Ratcatcher to exclude every strictly smaller width.
The hollow markers in Fig.~\ref{fig:depth-cost-and-width} show the resulting exact values $W^*$.
They coincide with the constructive width at every depth.

The Ratcatcher calculation does not directly provide the contraction path with minimum width.
A minimum-width decomposition can in principle be reconstructed using repeated Ratcatcher calls, for example by the cycle method of Hicks~\cite{hicksPlanarBranchDecompositions2005}.
Its worst-case running time scales as
\[
  O\!\left(|V|^4\right),
\]
where $|V|$ is the number of tensors.
For the $n=d=70$ network, this reconstruction is considerably less practical than directly generating the constructive sweep.

\section{From amplitudes to samples}
\label{sec:amplitudes-to-samples}
So far, we have considered tensor network methods for evaluating the amplitude associated with a single output bitstring.
Exact amplitudes are sufficient for validating experimental samples and for estimating fidelity using the cross-entropy benchmark (XEB)~\cite{boixoCharacterizingQuantumSupremacy2018,aruteQuantumSupremacyUsing2019}.

Amplitude evaluation alone, however, does not constitute a classical simulation of the sampling task itself.
Generating one output sample generally requires multiple amplitude evaluations and therefore incurs additional computational overhead.
For a fair comparison between classical and quantum sampling costs, we must account for the classical cost of producing one sample.

The random circuit sampling instances are expected to exhibit Porter--Thomas output statistics~\cite{boixoCharacterizingQuantumSupremacy2018,aruteQuantumSupremacyUsing2019}.
This structure enables several methods for converting amplitude evaluations into samples, including frugal rejection sampling~\cite{markovQuantumSupremacyBoth2018a,villalongaFlexibleHighperformanceSimulator2019,huangEfficientParallelizationTensor2021,kalachevClassicalSamplingRandom2021} and batched sampling~\cite{panSimulationQuantumCircuits2022b,panSolvingSamplingProblem2022a}.
In this section, we analyze the tensor network cost for generating samples using these methods.

\subsection{Frugal rejection sampling}
For an output bitstring $x\in\{0,1\}^n$, let
\begin{equation}
    p(x) = \left|\braket{x|U|0^n}\right|^2
\end{equation}
and define the rescaled probability
\begin{equation}
    z_x = 2^n p(x).
\end{equation}
For sufficiently scrambled random circuits, $z_x$ approximately follows the Porter--Thomas distribution~\cite{boixoCharacterizingQuantumSupremacy2018,aruteQuantumSupremacyUsing2019}:
\begin{equation}
    f(z) = e^{-z}, \qquad z\geq 0.
\end{equation}

Frugal rejection sampling converts these probability evaluations into output samples by introducing a cutoff parameter $M$~\cite{markovQuantumSupremacyBoth2018a,huangEfficientParallelizationTensor2021}.
It proceeds as follows:
\begin{enumerate}
    \item choose a candidate bitstring $x$ uniformly from $\{0,1\}^n$;
    \item evaluate $p(x)$ and compute $z_x=2^n p(x)$;
    \item accept $x$ with probability
    \begin{equation}
        \min\!\left(1, \frac{z_x}{M}\right).
    \end{equation}
    \item if $x$ is rejected, repeat the procedure.
\end{enumerate}

Under the Porter--Thomas statistics, the average acceptance probability is 
\begin{equation}
    \alpha(M)=
    \mathbb{E}_{z\sim\mathrm{Exp}(1)}\left[\min\!\left(1, \frac{z}{M}\right)\right] = \frac{1-e^{-M}}{M}.
\end{equation}
Consequently, the expected number of amplitude evaluations required to generate one sample is
\begin{equation}
    \frac{1}{\alpha(M)}=\frac{M}{1-e^{-M}}\simeq M.
\end{equation}
Let $q_M$ denote the normalized distribution produced by the accepted samples. Under the Porter--Thomas model, its total variation distance from the ideal output distribution satisfies
\[
\left\|q_M-p\right\|_{\mathrm{TV}}
=
\exp\!\left(
    -\frac{M}{1-e^{-M}}
\right)
=
O(e^{-M}).
\]
For example, $M=10$ gives
\[
\left\|q_M-p\right\|_{\mathrm{TV}}
\simeq 4.5\times10^{-5},
\]
which is accurate enough for our purpose.

\subsection{Batched sampling}
We next consider batched sampling based on a tensor network contraction with $k$ open output indices~\cite{panSimulationQuantumCircuits2022b,panSolvingSamplingProblem2022a,huangEfficientParallelizationTensor2021}.
We divide an output bitstring into a fixed upper-bit prefix $a$ and a lower-bit suffix $b$:
\begin{equation}
    x=(a,b),\qquad a\in\{0,1\}^{n-k},\qquad b \in\{0,1\}^k.
\end{equation}
Thus, $b$ corresponds to the output indices $q_{n-k+1},\ldots,q_n$.

For convenience, let
\begin{equation}
    A=2^{n-k},\qquad B=2^k.
\end{equation}
For each prefix $a$, define its marginal probability
\begin{equation}
    Q_a=\sum_b p(a,b).
\end{equation}
A batched sample is generated as follows:
\begin{enumerate}
    \item choose a prefix $a$ uniformly from $\{0,1\}^{n-k}$;
    \item contract the tensor network with $a$ fixed and the $k$ open indices, obtaining all $B=2^k$ amplitudes.
    \item compute the corresponding probabilities $p(a,b)$ and their sum $Q_a$;
    \item sample one suffix $b$ from the conditional distribution
    \begin{equation}
        p(b\mid a)=\frac{p(a,b)}{Q_a}.
    \end{equation}
\end{enumerate}
The resulting output distribution is
\begin{equation}
    q_k(a,b)=\frac{1}{A}\frac{p(a,b)}{Q_a}.
\end{equation}

In this procedure, we approximate the prefix distribution $Q_a$ by the uniform distribution.
The total variation distance is
\begin{align}
    \left\|q_k-p\right\|_{\mathrm{TV}}
    &=\frac12\sum_{a,b}\left|
        \frac{1}{A}\frac{p(a,b)}{Q_a} - p(a,b)
    \right| \\
    &= \frac12 \sum_a \left|Q_a-\frac{1}{A}\right| \label{eq:tv_norm}.
\end{align}

We assume a Porter--Thomas distribution.
Since $Q_a$ is the sum of $B=2^k$ approximately independent probabilities, its mean and variance are
\begin{equation}
  \mathbb{E}[Q_a]
  =
  \frac{1}{A},
  \qquad
  \operatorname{Var}(Q_a)
  \simeq
  \frac{1}{A^2B}.
\end{equation}
When $k$ is sufficiently large, the central limit theorem gives
\begin{equation}
  Q_a
  \approx
  \mathcal{N}\!\left(
      \frac{1}{A},
      \frac{1}{A^2B}
  \right).
\end{equation}
The mean absolute deviation of this normal distribution from $1/A$ is
\begin{equation}
  \mathbb{E}
  \left|
      Q_a-\frac{1}{A}
  \right|
  \simeq
  \frac{1}{A}
  \sqrt{\frac{2}{\pi B}}.
\end{equation}
Substituting this result into Eq.~\eqref{eq:tv_norm} gives
\begin{equation}
  \mathbb{E}
  \left\|q_k-p\right\|_{\mathrm{TV}}
  \simeq
  \frac{1}{\sqrt{2\pi B}}
  =
  \frac{1}{\sqrt{2\pi\,2^k}}.
\end{equation}
For $k=8$, one contraction evaluates a batch of $B=256$ amplitudes, and the error estimate gives
\begin{equation}
  \mathbb{E}
  \left\|q_8-p\right\|_{\mathrm{TV}}
  \simeq
  2.5\times10^{-2}.
\end{equation}
We can further reduce the sampling error by combining with frugal rejection sampling~\cite{huangEfficientParallelizationTensor2021}.
For simplicity, however, we focus on the basic batched sampler described above.

\subsection{Constructive path for batched sampling}

We now construct a tensor network contraction path that leaves $k$ output indices open and thereby evaluates the $B=2^k$ amplitudes required for batched sampling in a single contraction.
If its cost remains comparable to that of computing a single amplitude ($k=0$), batched sampling can be more efficient than frugal rejection sampling.
Unlike single-amplitude contraction, however, the final result is a rank-$k$ tensor rather than a scalar.
The contraction path must therefore be chosen carefully to avoid excessively large intermediate tensors arising from the accumulation of open output indices.

\begin{figure}[t]
  \centering
  \includegraphics[width=\linewidth]{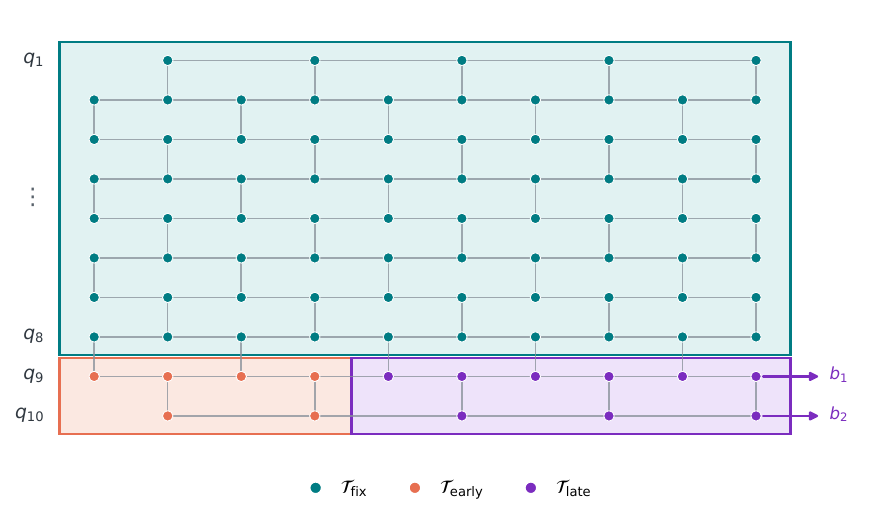}
  \caption{
      Three-region decomposition for batched contraction with $n=d=10$ and $k=2$.
      The fixed region $\mathcal{T}_{\mathrm{fix}}$ spans qubits $q_1,\ldots,q_8$, while the output indices on $q_9$ and $q_{10}$ are left open.
      The early and late regions contain the first $c=4$ layers and remaining $d-c=6$ entangling layers, respectively.
  }
  \label{fig:batched-regions}
\end{figure}

We restrict $k$ to the range
\begin{equation}
  1\leq k\leq\left\lfloor\frac{d}{4}\right\rfloor.
\end{equation}
We partition the tensor network into three regions, as illustrated in Fig.~\ref{fig:batched-regions}.
Let
\begin{equation}
  c=2k.
\end{equation}
The fixed region $\mathcal{T}_{\mathrm{fix}}$ contains all tensors associated with qubits $q_1,\ldots,q_{n-k}$.
Within the open-output band consisting of qubits $q_{n-k+1},\ldots,q_n$, the early region $\mathcal{T}_{\mathrm{early}}$ contains the tensors in the first $c$ entangling layers, while the late region $\mathcal{T}_{\mathrm{late}}$ contains the remaining tensors.

The fixed region is contracted by the sweep of Sec.~\ref{sec:temporal-sweep}, starting at $q_1$.
The early region is contracted independently from the $q_n$ boundary in the opposite spatial direction.
After joining the two resulting boundary tensors, we absorb the remaining tensors in causal
layer order.
The resulting deterministic contraction procedure is summarized in Algorithm~\ref{alg:batched-contraction}.

\begin{algorithm}[t]
    \caption{Deterministic temporal-state contraction for a batch of exact amplitudes}
    \label{alg:batched-contraction}
    \begin{algorithmic}[1]
    \Require A brickwork tensor network $\mathcal{T}_a$ with prefix $a\in\{0,1\}^{n-k}$ fixed and suffix $b=(b_1,\ldots,b_k)$ open
    \Require $1\leq k\leq\lfloor d/4\rfloor$
    \Ensure A rank-$k$ tensor $Z$ satisfying $Z[b]=\braket{a,b|U|0^n}$ for every $b\in\{0,1\}^k$

    \State $h\gets\lceil d/2\rceil$
    \State $c\gets2k$
    \State $(\mathcal{T}_{\mathrm{fix}},
            \mathcal{T}_{\mathrm{early}},
            \mathcal{T}_{\mathrm{late}})
            \gets\Call{Partition}{\mathcal{T}_a,c}$

    \State $Z_{\mathrm{fix}}\gets
    \Call{ContractSweep}{\mathcal{T}_{\mathrm{fix}},h}$
    \Comment{$q_1$ to $q_{n-k}$; reverse time on even-numbered rows}
    
    \State $Z_{\mathrm{early}}\gets
    \Call{ContractSweep}{\mathcal{T}_{\mathrm{early}},h}$
    \Comment{$q_n$ to $q_{n-k+1}$;}

    \State $Z\gets\Call{Contract}{Z_{\mathrm{fix}},Z_{\mathrm{early}}}$

    \For{$\ell\gets c+1$ to $d$}
      \State $Z\gets\Call{ContractLayer}{
          Z,\mathcal{T}_{\mathrm{late}}^{(\ell)}}$
    \EndFor

    \State \Return $Z$
    \end{algorithmic}
\end{algorithm}

\begin{proposition}[Batched-contraction width]
Consider an $n$-qubit bond-dimension-$2$ open-boundary brickwork tensor network with
\begin{equation}
  n\geq4,\qquad 7\leq d\leq2n,\qquad
  1\leq k\leq\left\lfloor\frac d4\right\rfloor.
\end{equation}
The contraction order in Algorithm~\ref{alg:batched-contraction} produces an unsliced contraction path $P_{\mathrm{con-batch}}$ whose contraction width is
\begin{equation}
    W(P_{\mathrm{con-batch}})=\lceil d/2\rceil.
\end{equation}
\end{proposition}

\begin{proof}
Let $h=\lceil d/2\rceil$.
The fixed region $\mathcal{T}_{\mathrm{fix}}$ is contracted using the sweep of Algorithm~\ref{alg:constructive-sweep}, and therefore produces no intermediate tensor of rank greater than $h$.
The same analysis applies to the reverse sweep of $\mathcal{T}_{\mathrm{early}}$.
Its temporal depth is $c=2k$, and $k$ open edges connect $\mathcal{T}_{\mathrm{early}}$ to $\mathcal{T}_{\mathrm{late}}$.
Therefore, $\mathcal{T}_{\mathrm{early}}$ is also constructed using a similar algorithm with contraction width at most $\max\{2k,4\}\leq h$.

The first $c=2k$ alternating entangling layers cross the boundary between the fixed and open-output regions $k$ times.
Thus, $Z_{\mathrm{fix}}$ and $Z_{\mathrm{early}}$ share $k$ indices. 
Moreover, $Z_{\mathrm{early}}$ has exactly $k$ indices leading to $\mathcal{T}_{\mathrm{late}}$.
Therefore,
\begin{equation}
    \operatorname{rank}Z
    =\operatorname{rank}Z_{\mathrm{fix}}+(k+k)-2k
    =\operatorname{rank}Z_{\mathrm{fix}}
    \leq h.
\end{equation}

It remains to contract layers $\ell=c+1,\ldots,d$.
Within each layer, a complete CZ produces a rank-four tensor that shares two indices with $Z$, so its absorption does not increase the rank of $Z$.
A CZ crossing $\mathcal{T}_{\mathrm{fix}}$ and $\mathcal{T}_{\mathrm{late}}$ boundary contributes a rank-three tensor sharing two indices with $Z$, so its absorption decreases the rank.
\end{proof}

\begin{figure}[t]
    \centering
    \includegraphics[width=\linewidth]{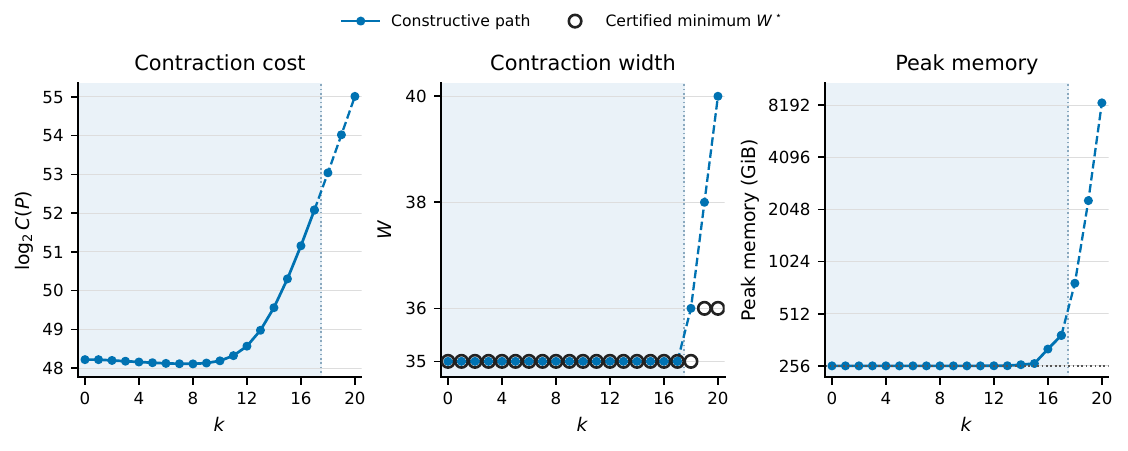}
    \caption{
      Resource scaling with the number of open output indices $k$ for the $n=d=70$ brickwork circuit.
      Filled blue circles show the deterministic constructive path in Algorithm~\ref{alg:batched-contraction}, and black open circles show the certified minimum contraction width $W^\star$ obtained by Ratcon.
      The shaded region, $0\leq k\leq17$, is the proved domain of the contraction width.
      The right panel reports the modeled complex64 tensor payload, not a measured implementation peak allocation.
    }
    \label{fig:open-k-resources}
\end{figure}

Fig.~\ref{fig:open-k-resources} shows the contraction resources as functions of $k$ for $n=d=70$.
For the contraction width, we also overlay the minimal width obtained by Ratcon.
As shown in the figure, we can see that our constructive path achieves the minimal contraction width for $k\le\lfloor d/4\rfloor$.
Moreover, the total contraction cost changes little for small $k$, so obtaining a batch of amplitudes introduces little arithmetic overhead relative to a single amplitude.
This is also the case for peak size; the live tensor payload remains nearly constant for small $k$.
By appropriately choosing $k$, we can simulate the sampling task without large space and time overhead compared to the amplitude evaluation.

\section{Numerical results}
\label{sec:numerical-results}

We implemented the deterministic contraction algorithm on GPUs and evaluated its runtime, memory consumption, and numerical accuracy.
The single-GPU implementation uses the cuStateVec component of the NVIDIA cuQuantum SDK~\cite{bayraktarCuQuantumSDKHighPerformance2023}.
The multi-GPU implementation follows the distributed tensor network contraction framework of Ref.~\cite{panParallelizingLargeScaleTensor2026} and uses NVIDIA cuTENSORMp~\cite{MultiProcessSupportCuTENSORMp} for distributed tensor contractions.

\subsection{Single-GPU experiment}
\label{sec:single-gpu-experiment}

We first benchmark the $n=70$, bond-dimension-$2$, open-boundary brickwork tensor networks on a 16-GB NVIDIA GeForce RTX 5060 Ti.
We fix the number of open indices $k=8$ and vary the circuit depth from $d=32$ to $60$.

\begin{figure}[t]
  \centering
  \includegraphics[width=\linewidth]
  {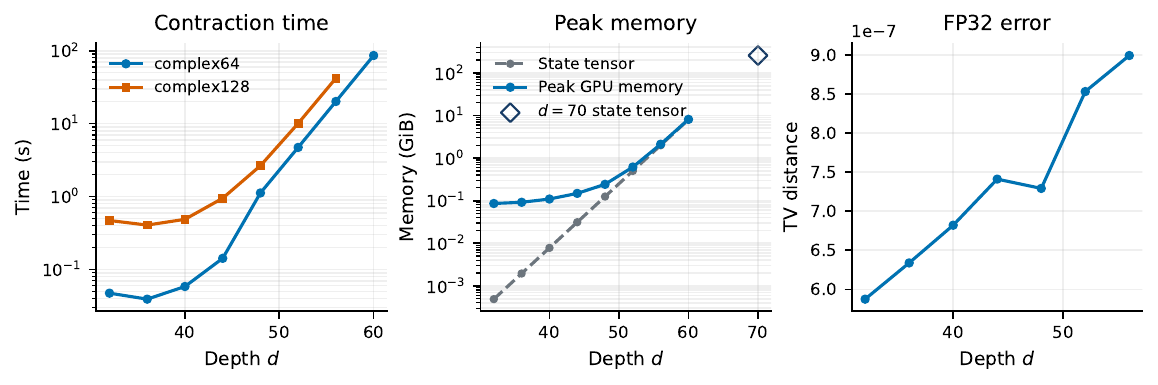}
  \caption{
      Single-GPU execution for the $n=70$ and $k=8$ contractions.
      The left and center panels show the contraction time and the peak device-memory allocation above the initialized baseline, respectively.
      The open diamond marks the analytically determined $256$-GiB size of the largest intermediate tensor at $d=70$.
      The right panel shows the total variation distance between the conditional probability distributions obtained with complex64 and complex128 arithmetic.
  }
  \label{fig:gpu-depth-scaling}
\end{figure}

The measured contraction time and peak device-memory allocation are shown in Fig.~\ref{fig:gpu-depth-scaling}.
At small depths, the runtime is dominated by fixed library and kernel-launch overheads.
Once the boundary tensor becomes sufficiently large, both the runtime and the largest intermediate tensor size exhibit the expected exponential scaling with $\lceil d/2\rceil$.
For example, at $d=60$, the complex64 temporal-state tensor has an $8$~GiB payload, while the measured increase in device-memory allocation above the initialized baseline is $8.117$~GiB.
Including the $1.126$~GiB initialization baseline, the measured total peak allocation is $9.243$~GiB.
In this regime, the state tensor grows exponentially, whereas these fixed and auxiliary allocations become negligible in comparison.

For the depths at which both complex64 and complex128 calculations fit in memory, we compare their conditional probability distributions.
Up to $d=56$, their total variation distance remains below $9.0\times10^{-7}$.
These measured discrepancies are negligible compared with the sampling uncertainty considered in this work and support the use of complex64 arithmetic for the full-depth calculation.

\subsection{Multi-GPU simulation of the IBM circuit}
\label{sec:multi-gpu-experiment}

For $n=d=70$, the contraction width is $35$, and hence the largest intermediate tensor contains $2^{35}$ entries, corresponding to $256$~GiB.
We distribute this tensor over eight NVIDIA H100 Tensor Core GPUs.
The $256$~GiB value is the largest-intermediate tensor payload, rather than a measurement of total distributed peak allocation; workspaces, communication buffers, and other implementation allocations are additional.
Local updates are executed independently on each shard, while an update acting on a distributed index requires an inter-GPU index exchange.
The production driver merges groups of small branch tensors before contracting them with the large boundary tensor.
This slight change of the contraction path improves GPU efficiency without changing the requested output tensor or the maximum intermediate of $2^{35}$ entries.

The 2,051 amplitude batches shared the same tensor network topology, tensor dimensions, and contraction schedule, differing only in the fixed output-bit assignments.
We divided them among 32 independent single-node jobs: 31 jobs processed 64 configurations each, and the final job processed 67.
Within each job, the same eight MPI ranks, one per GPU, processed the assigned configurations sequentially.
Each contraction returned a vector of $2^8=256$ amplitudes.
The first contraction in each job initialized the per-rank cuTENSORMp handle and populated the contraction plan and workspace caches, and averaged 54.63~s across the 32 jobs.
These resources were reused for the remaining 2,019 contractions, which averaged 32.16~s.
The successful jobs averaged 35~min 48~s, and the longest -- the 67-configuration job -- took 37~min 16~s. Since the jobs were mutually independent, this maximum job time gives a computational makespan of 37~min 16~s when all 32 nodes are available concurrently.

Following the standard fidelity-weighted resource accounting used in classical simulations of RCS, a sample set containing $N$ outputs at fidelity $F$ is assigned the same sampling cost as $NF$ exact samples, with the remaining outputs drawn uniformly~\cite{villalongaEstablishingQuantumSupremacy2019,huangEfficientParallelizationTensor2021}.
For the $2051$ experimental outputs and the reported fidelity lower bound $F=0.284$, the corresponding classical workload is therefore
\begin{equation}
  N_{\mathrm{exact}}
  =\left\lceil
    2051\times0.284
  \right\rceil
  =583
\end{equation}
exact samples, supplemented with 1468 uniformly random outputs to form 2051 bitstrings.
583 contractions require approximately 5.21~h on one eight-H100 node.
Distributing them over 32 such nodes gives an estimated wall time of approximately 10.6~min.

We additionally use the calculated amplitudes to evaluate cross-entropy benchmarking on the published experimental bitstrings.
For each experimental bitstring $x_i$, we select the corresponding component of its amplitude batch and evaluate
\begin{equation}
  p(x_i)
  =
  \left|\braket{x_i|U|0^n}\right|^2.
\end{equation}
We define the log-XEB estimator as
\begin{equation}
  \widehat{F}_{\log\mathrm{XEB}}
  =
  \gamma+n\ln 2
  +\frac{1}{N}\sum_{i=1}^{N}\ln p(x_i),
  \label{eq:normalized-log-xeb}
\end{equation}
where $\gamma$ is the Euler constant and $N=2051$.
Under the Porter--Thomas and sufficiently scrambled-noise assumptions conventionally used in random-circuit sampling, this log-XEB statistic estimates the circuit fidelity~\cite{boixoCharacterizingQuantumSupremacy2018,aruteQuantumSupremacyUsing2019}.

\begin{figure}[t]
  \centering
  \includegraphics[width=0.72\linewidth]
  {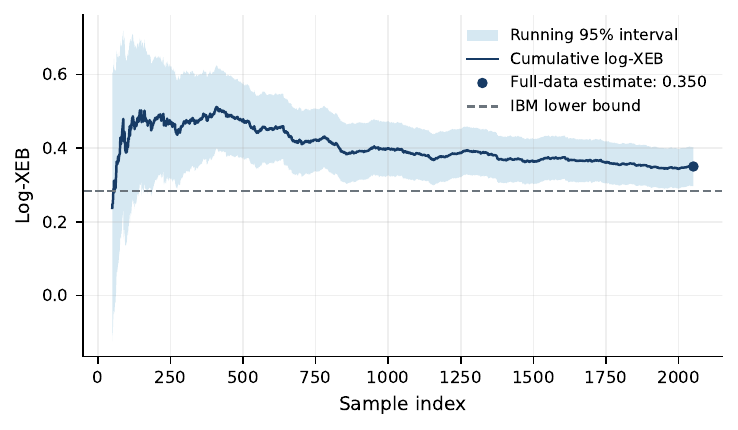}
  \caption{
      Log-XEB estimate for the 2051 experimental bitstrings.
      The blue line and shaded region show the estimate and pointwise 95\% intervals obtained in the archived row order.
      The dashed line shows IBM's $0.284$ state-fidelity lower bound.
  }
  \label{fig:full-log-xeb}
\end{figure}

For the 468-$T$ experimental dataset, we obtain
\begin{equation}
  \widehat{F}_{\log\mathrm{XEB}}
  =0.35034,
  \qquad
  \mathrm{SE}=0.02689.
\end{equation}
Here $0.02689$ is the empirical standard error of the sample mean when the 2051 bitstrings are treated as independent draws.
The corresponding 95\% interval is
\begin{equation}
  [0.29763,\,0.40305].
\end{equation}
Figure~\ref{fig:full-log-xeb} shows the convergence of the estimator as the experimental bitstrings are accumulated.

IBM's reported value $0.284$ is instead a one-sided 95\%-confidence lower bound on the state fidelity derived from its direct fidelity estimation and syndrome analysis~\cite{martielSamplingHardCircuits2026}.
Under the standard scrambled-noise interpretation, our log-XEB result provides an independent fidelity proxy that is numerically compatible with this lower bound.

Table~\ref{tab:full-depth-results} summarizes the multi-GPU computational results together with the sampling-cost projection derived above.

\begin{table}[t]
  \centering
  \caption{Summary of the measured $T=468$ doped Clifford circuit simulation and the sampling-cost projection.}
  \label{tab:full-depth-results}
  \begin{tabular}{lr}
    \toprule
    Quantity & Value \\
    \midrule
    Circuit size $(n,d)$ & $(70,70)$ \\
    Open outputs $k$ & $8$ \\
    Amplitudes per contraction & $256$ \\
    GPUs per contraction & $8$ NVIDIA H100 GPUs \\
    Production jobs & $32$ single-node jobs \\
    Largest-intermediate payload (complex64) & $256$ GiB \\
    Validated experimental bitstrings & $2051$ \\
    Mean first-use contraction time & $54.63$ s \\
    Mean steady-state contraction time & $32.16$ s \\
    Longest job / computational wall time & $37$ min $16$ s \\
    log-XEB & $0.35034$ \\
    log-XEB 95\% interval & $[0.29763,0.40305]$ \\
    Batched contractions for the $0.284$ proxy & $583$ \\
    Projected serial contraction time & $5.21$ h \\
    Projected contraction time on 32 nodes & $10.6$ min \\
    \bottomrule
  \end{tabular}
\end{table}

\section{Implications for circuit design}
\label{sec:circuit-design}

The preceding results identify the structural features that make the present circuit family accessible to exact tensor network simulation. 
Viewed from the opposite direction, these results constrain the design of quantum-advantage experiments.
A candidate circuit must make the corresponding classical simulation impractical while retaining sufficient hardware fidelity and efficient verification.
This section discusses the resulting constraints in circuit design.

\subsection{Slicing overhead}
\label{sec:slicing-overhead}

Let $h=\lceil d/2\rceil$. 
The unsliced temporal-state contraction stores a boundary tensor with $2^h$ complex entries. 
If this tensor does not fit in aggregate memory, one may think of slicing indices and evaluate the resulting subproblems independently. 

To quantify this tradeoff, we consider the circuit instances with $n=d$ and start from the constructive contraction tree of width $W=h$. 
We then apply cotengra's \texttt{slice\_and\_reconfigure} routine with the forced target
\begin{equation}
  W_{\mathrm{target}}=h-1.
\end{equation}

\begin{figure}[t]
  \centering
  \includegraphics[width=\linewidth]
  {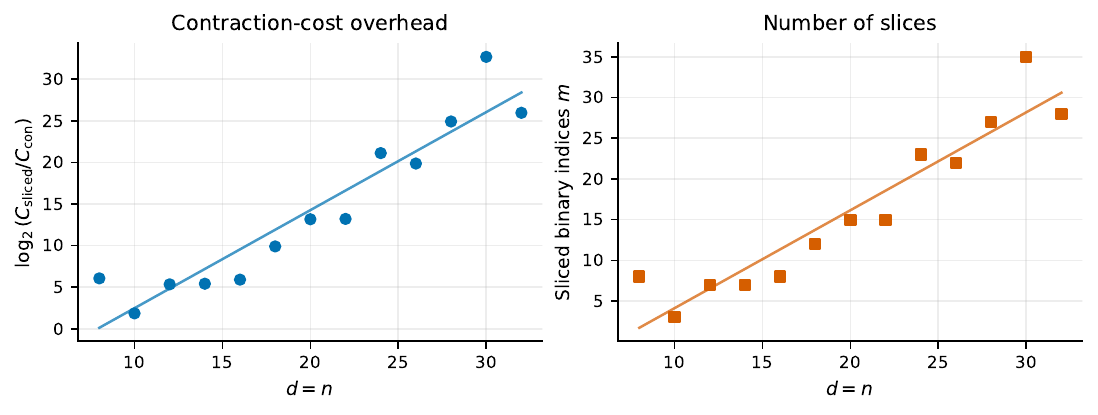}
  \caption{
      Slicing overhead for brickwork tensor networks with $n=d$.
      The target width is one bit below the unsliced constructive width, $W_{\mathrm{target}}=h-1$.
      The left panel shows the contraction cost ratio relative to the unsliced constructive path, and the right panel shows the number of independent slice tasks.
      Lines are guides to the eye.
  }
  \label{fig:slicing-overhead}
\end{figure}

Figure~\ref{fig:slicing-overhead} shows the slicing overhead and the number $m$ of sliced binary indices.
As shown in the figure, the penalty grows rapidly with circuit size, and the required number of independent tasks, $2^m$, increases exponentially.
Reducing the main state tensor by only a factor of two does not merely double the runtime for this geometry: it can require a very large number of repeated contractions.
This result is empirical but nevertheless indicates that access to $O(2^h)$ aggregate memory is the essential resource for executing the exact tensor network contraction with moderate arithmetic overhead. 
In this practical sense, memory capacity, rather than the contraction cost alone, is the main bottleneck.

This observation has a direct implication for circuit depth.
Because slicing does not efficiently reduce the natural contraction width, increasing $d$ raises the practical memory requirement essentially as $2^{\lceil d/2\rceil}$.
For reference, Pan \emph{et al.} performed distributed tensor network contractions using as many as 1024 H100 GPUs across multiple nodes~\cite{panParallelizingLargeScaleTensor2026}.
This configuration provides approximately 80~TB of aggregate GPU memory.
Even under the optimistic assumption that the entire capacity could be used for a single complex64 intermediate tensor, it could store only approximately $2^{43}$ entries, corresponding to a circuit depth of at most $d\simeq86$ in the present family.

Increasing the depth, however, places additional demands on the experiment.
At fixed $n$, the number of CZ gates grows approximately as $d(n-1)/2$, increasing the accumulated physical error.
Moreover, it remains unclear whether deeper circuits admit low-overhead spacetime-code checks with sufficiently large detecting regions and suitable stabilizer-compatible doping locations under the connectivity and scheduling constraints of the target hardware architecture.
A depth-based circuit design must therefore be evaluated jointly in terms of tensor network cost, physical error, postselection acceptance, check coverage, and T-gate capacity.

\subsection{Effect of bond dimension}
\label{sec:bond-dimension-effect}

The favorable scaling of the temporal-state contraction relies on the bond dimension $\chi=2$ of the CZ decomposition. 
A temporal cut crossing $d/2$ entangling bonds carries $\chi^{d/2}$ components.
Thus, increasing the bond dimension from $\chi=2$ to $\chi=4$ doubles the leading width from $d/2$ to $d$.

We construct a path for $\chi=4$ by adapting the blocked spatial sweep used for $\chi=2$.
This construction gives
\begin{equation}
    W\left(P_{\mathrm{con}}^{(\chi=4)}\right)
    =
    \begin{cases}
      \min(n,d-1), & d \text{ even},\\
      \min(n,d),   & d \text{ odd}.
    \end{cases}
\end{equation}
The total contraction cost is
\begin{equation}
C\left(P_{\mathrm{con}}^{(\chi=4)}\right)
= O\!\left(nd\,2^d\right).
\end{equation}
At every depth for which an exact Ratcatcher calculation was performed, the certified minimum width agrees with the constructive width.

\begin{figure}[t]
  \centering
  \includegraphics[width=\linewidth]
  {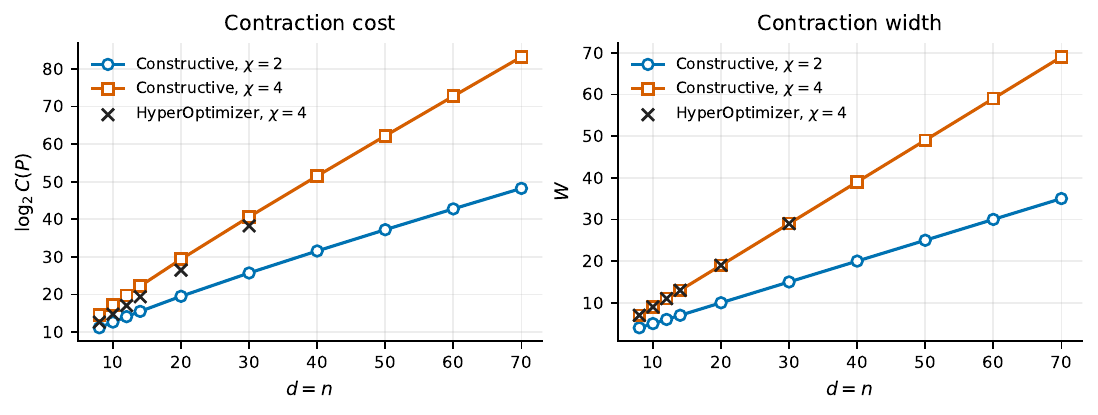}
  \caption{
      Effect of the entangling-bond dimension on the total contraction cost (left) and contraction width (right) for circuits with $n=d$.
      The blue and orange lines show the constructive sweep for $\chi=2$ and $\chi=4$, respectively.
      Black crosses show paths returned by HyperOptimizer for $\chi=4$ at small $d$ using \texttt{minimize="flops"}.
  }
  \label{fig:bond-dimension-cost}
\end{figure}

Figure~\ref{fig:bond-dimension-cost} compares the constructive paths for $\chi=2$ and $\chi=4$.
At $n=d=70$, increasing the bond dimension from two to four changes the path estimates from
\begin{equation}
    (W,\log_2 C)=(35,48.22)
    \quad\text{to}\quad
    (69,83.19).
\end{equation}
The largest intermediate in the $\chi=4$ case consequently contains $2^{69}$ complex entries, which is intractable for the current devices.
The small-size HyperOptimizer results reduce the constants slightly but retain the doubled width slope.

Replacing CZ by a generic bond-dimension-$4$ two-qubit gate is therefore an effective way of hardening the circuit against this temporal-state simulator.
This change, however, is not neutral for the experiment.  A generic $2$-qubit gate is non-Clifford, so Pauli checks no longer propagate as Pauli operators and the present spacetime-code construction does not apply
directly. 
A bond-dimension-$4$ Clifford entangler would preserve Pauli propagation in principle, but it may require several native entangling gates if it is not implemented directly by the hardware. 

\subsection{Effect of boundary geometry}
\label{sec:boundary-geometry-effect}

The open end of the qubit chain is what allows the temporal boundary state to be initialized and swept across the network. 
Adding an interaction between the two ends changes the tensor network from a strip into a cylinder and removes this starting boundary.

For the ring geometry, we construct a path using the same idea as for the open chain.
This gives
\begin{equation}
W\left(P_{\mathrm{con}}^{(\mathrm{ring})}\right)
=\min(n,d).
\end{equation}
The corresponding contraction cost is
\begin{equation}
C\left(P_{\mathrm{con}}^{(\mathrm{ring})}\right)
=O\!\left(nd\,2^{\min(n,d)}\right).
\end{equation}
For every tested even pair $(n,d)$, we confirmed that the constructive width agrees with the minimum one obtained by Ratcatcher.

However, the ring geometry also admits a simple structure-aware slicing scheme.
For instances with $n=d$, slicing all $d/2$ bonds associated with the endpoint CZ gates produces $2^{d/2}$ independent 1D chain tasks, each with width $d/2$.
The total cost over all slice tasks retains the same leading scaling as the unsliced ring contraction,
\begin{equation}
W\left(P_{\mathrm{slice}}^{(\mathrm{ring})}\right)=\frac{d}{2},
\qquad
C\left(P_{\mathrm{slice}}^{(\mathrm{ring})}\right)
=O\!\left(nd\,2^d\right).
\end{equation}

\begin{figure}[t]
  \centering
  \includegraphics[width=\linewidth]
  {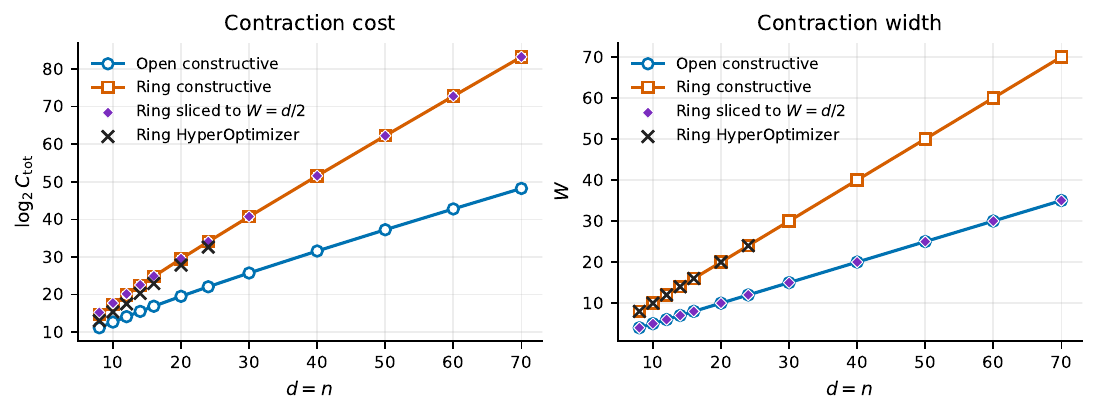}
  \caption{
      Effect of closing the spatial boundary on the total contraction cost (left) and contraction width (right) for bond-dimension-$2$ brickwork tensor networks with $n=d$.
      The blue and orange lines show the constructive paths for the open and ring geometries, respectively.
      The purple diamond shows the ring path sliced to $W=d/2$.
      The black crosses show paths returned by HyperOptimizer for the ring geometry at small $d$ using \texttt{minimize="flops"}.
  }
  \label{fig:boundary-geometry-cost}
\end{figure}

Figure~\ref{fig:boundary-geometry-cost} compares the open, ring, and sliced-ring constructions.
At $n=d=70$, their path estimates are
\begin{equation}
    (W,\log_2 C)_{\mathrm{open}}=(35,48.22),
\end{equation}
and
\begin{equation}
    (W,\log_2 C)_{\mathrm{ring}}=(70,83.22),
    \qquad
    (W,\log_2 C)_{\mathrm{ring,sliced}}=(35,83.22).
\end{equation}
Thus, slicing reduces the peak width of the ring contraction to that of the open-chain path, but leaves a $2^{35}$ difference in their contraction costs.

Closing the chain is therefore attractive from the viewpoint of classical hardness, provided that the endpoint coupling is native to the hardware.
For the brickwork pattern, closing the chain requires only one additional CZ gate in every other
entangling layer, or $O(d)$ additional gates in total.
It can increase the number of spacetime locations reached by back-propagated Pauli operators.
Determining whether these advantages can be combined with low-overhead spacetime-code checks under the constraints of the target hardware is an important direction for future work.

\subsection{Classical simulation frontier}
\label{sec:classical-simulation-frontier}

\begin{figure}[t]
    \centering
    \includegraphics[width=0.9\linewidth]
    {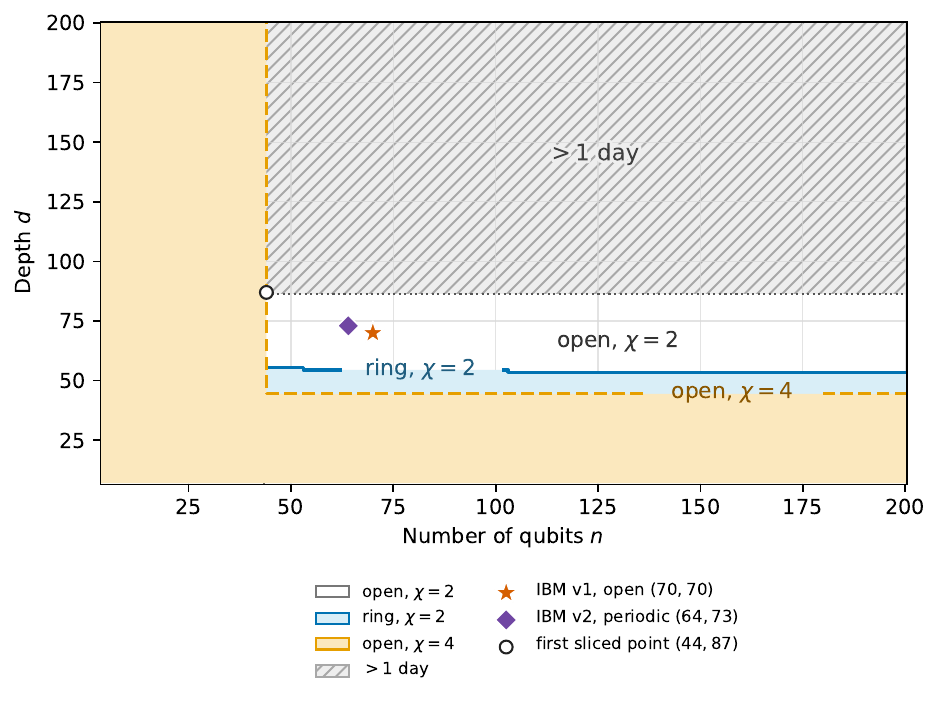}
    \caption{
        Classical simulatability map for computing one exact amplitude of the one-dimensional brickwork circuit.
        Each circuit is evaluated using the available deterministic paths and applicable slicing strategies.
        Colors indicate the circuit family that remains within the one-day limit: the open chain with $\chi=4$ (orange), the ring with $\chi=2$ (blue), or the open chain with $\chi=2$ (white).
        Gray hatching indicates a modeled runtime greater than one day.
        The orange star marks the open-boundary IBM v1 circuit; the purple diamond marks the periodic-boundary IBM v2 circuit. The open circle marks the first open $\chi=2$ instance that requires slicing.
    }
    \label{fig:classical-quantum-crossover}
\end{figure}

To summarize these resource estimates, Fig.~\ref{fig:classical-quantum-crossover} shows the circuit region in which one exact amplitude can be classically simulated within one day.
We assume 1024 H100 GPUs and aggregate memory for $2^{43}$ complex64 entries.
For contractions distributed over all 1024 GPUs, we apply the 10\% retained per-GPU throughput inferred from the 1024-GPU scaling reported by Pan \emph{et al.}~\cite{panParallelizingLargeScaleTensor2026}.
Applying this factor to our measured single-node rate of 17.68~TFLOP/s per H100 gives an effective rate of 1.81~PFLOP/s.
For the sliced ring contraction, each independent slice task fits on a single H100 GPU, so we use the measured rate of 17.68~TFLOP/s per H100 without inter-node tensor communication.

Within the plotted range, all circuit families are feasible when $n\leq 43$ and $d\leq 43$.
For large $n$, the open $\chi=4$ circuit exceeds the limit for $d>44$.
The $\chi=2$ ring remains feasible through approximately $d=53$ because of the structure-aware slicing strategy, but eventually exceeds the one-day limit due to its exponentially increasing contraction cost.
The open $\chi=2$ circuit is simulatable through $d=86$.
The white circle at $(n,d)=(44,87)$ marks the first open $\chi=2$ instance that requires slicing to fit in memory.
Even for this smallest instance, cotengra's \texttt{slice\_and\_reconfigure} slices 18 indices and produces $2^{18}$ independent subtasks, resulting in a modeled runtime of approximately 739 days.

The star marks the initial IBM v1 instance at $(n,d)=(70,70)$, which remains feasible for the open $\chi=2$ circuit.
The diamond marks the revised IBM v2 instance.
The revised IBM circuit (IBM v2) has $n=64$ qubits on a one-dimensional ring, $d=73$ entangling layers, and $314$ inserted $T$ gates~\cite{martielSamplingHardCircuits2026}.
The point $(n,d)=(64,73)$ lies beyond the classically simulatable frontier for the ring $\chi=2$ circuit.
For one exact output probability of this revised circuit, an explicit path sliced to $W\leq43$ has aggregate contraction cost $\log_{10}C=25.93$.

The three circuit families are limited in different ways.
For the open $\chi=2$ circuit, the memory threshold $W=\lceil d/2\rceil$ sharply separates the feasible and infeasible regions.
For the open $\chi=4$ circuit, the corresponding depth limit is halved.
For the ring, the slicing strategy removes the per-task memory bottleneck, but the total FLOPs still produces a clear frontier at a much smaller depth than for the open $\chi=2$ circuit.
These classical limits should be considered together with the hardware resources available to the experiment and the requirements of the spacetime code when designing circuits capable of demonstrating quantum advantage.

\section{Discussion}
\label{sec:discussion}

In this work, we evaluated exact ideal amplitudes and probabilities for all 2051 published output bitstrings of IBM's largest doped Clifford sampling instance.
Our deterministic temporal-state contraction has width $\lceil d/2\rceil$ for the stated open-boundary, bond-dimension-$2$ brickwork family and retains this width when $k\leq\lfloor d/4\rfloor$ output bits are left open.
For the $n=d=70$ instance this gives width $35$ and a 256-GiB largest-intermediate complex64 tensor payload.
Ratcatcher calculations certify the minimum width for the $n=70$ instances at $d=10,20,\ldots,70$ shown in Fig.~\ref{fig:depth-cost-and-width}.

The distributed implementation evaluated all 2051 published-output batches in a measured makespan of 37.3 minutes using 32 eight-H100 nodes.
This execution result is distinct from the fidelity-weighted resource projection: assigning 583 contractions to the corresponding workload gives an estimated makespan of 10.6 minutes on the same resources.
The calculated probabilities give a log-XEB estimate of $0.35034$ with a 95\% interval $[0.29763,0.40305]$ when the bitstrings are treated as independent draws.
Under the Porter--Thomas and scrambled-noise assumptions used in the fidelity interpretation of log-XEB, this result is numerically compatible with IBM's separately obtained fidelity lower bound.

Evolving a temporal state in the spatial direction is a standard technique in the tensor network community~\cite{banulsMatrixProductStates2009a,leroseInfluenceMatrixApproach2021}.
The conceptual reason for the favorable contraction here is distinct from the circuit's Clifford-versus-non-Clifford composition.
After the one-qubit gates are absorbed, the dense contraction schedule is governed by the open quasi-one-dimensional geometry, circuit depth, and operator Schmidt rank of the entangling gates.
Its width and modeled arithmetic cost are therefore independent of the number and placement of the $T$ gates and of the values of the one-qubit gates.
This explains how a circuit deliberately made difficult for near-Clifford and entanglement-based methods can nevertheless possess a low-width tensor network path.
Generic path search does not necessarily expose this regularity, motivating path optimizers that incorporate known geometric sweeps as algorithmic priors.

Our results also show that memory capacity is a crucial resource in practical tensor network simulation.
Setting the target contraction width only one below the unsliced minimum results in substantial slicing overhead and a much larger total contraction cost.
An alternative to slicing is to distribute the dominant intermediate tensor across multiple GPUs and eventually across multiple nodes~\cite{fuSurpassingSycamoreAchieving2024,zhaoLeapfroggingSycamoreHarnessing2025}.
Future circuit simulators should combine structure-aware path optimization with tensor distribution while accounting for arithmetic cost, memory capacity, slicing overhead, and communication.
Similar strategies may also benefit other large-scale tensor network calculations.

Classical simulation has repeatedly provided quantitative benchmarks for assessing claims of quantum advantage~\cite{huangEfficientParallelizationTensor2021,oh_Tensor_2023,tindallEfficientTensorNetwork2024,tindallDynamicsDisorderedQuantum2025,rauschPushingClassicalFrontier2026}.
The exact probabilities calculated here provide an independent diagnostic of IBM's published output data.
In addition to the large-scale execution itself, the quantitative log-XEB result makes this diagnostic directly comparable with the fidelity scale reported by IBM, provided that the assumptions used in interpreting log-XEB are kept explicit.

Our bond-dimension-$4$, periodic-boundary, and slicing analyses extend the result from one experimental instance to modeled design sensitivities for related brickwork families.
They show that depth, operator Schmidt rank, boundary geometry, and available memory can change the cost of the analyzed exact-amplitude strategy by many orders of magnitude.
These circuit changes can simultaneously affect physical error rates, native-gate overhead, and the construction of effective spacetime checks.
Classical hardness and experimental verifiability should therefore be considered together when designing future sampling experiments.
In this role, structure-aware tensor networks provide both a post-experiment diagnostic and a quantitative tool for circuit co-design.

\section*{Data availability}

The amplitudes and contraction paths are available on Zenodo at \url{https://doi.org/10.5281/zenodo.21912448}.

\section*{Note added}
We initially analyzed only the 70-qubit, depth-70, 468-$T$-gate, open-boundary 1D doped Clifford circuit considered in the first arXiv version of Ref.~\cite{martielSamplingHardCircuits2026}.
The experimental design was revised to a 64-qubit, depth-73, 314-$T$-gate 1D doped Clifford circuit with periodic boundary conditions.
We subsequently analyzed the classical cost of simulating the new circuit, and it appears to be beyond the practical reach of the simulation method used in this work.

\section*{Acknowledgements}

This work is supported by the Ministry of Education Singapore under grant No. SKI 2021\_07\_03,  National Quantum Computing Hub translational fund No. W24Q3D0002, MTC Young Individual Research Grant No. M25N8c0120 and Advanced Quantum Algorithms and Solutions (AQAS) Grant No. S25Q9DA001/02.

\printbibliography

\end{document}